\documentclass[10pt,twoside]{article}
\date{}

\usepackage{amsfonts}
\usepackage{amsmath}
\usepackage{amssymb}
\usepackage{latexsym}
\usepackage{fancyhdr}
\usepackage{algorithm}
\usepackage{algpseudocode}
\usepackage{theorem}
\usepackage[margin=2.5cm]{geometry}
\usepackage{graphicx,xcolor}

\theoremstyle{break}
\newtheorem{thm}{Theorem}
\newtheorem{cor}{Corollary}
\newtheorem{lem}{Lemma}
\newtheorem{prop}{Proposition}
\newtheorem{defn}{Definition}

\newtheorem{exmp}{Example}

\catcode`\ç=13
\defç{\c{c}}
\catcode`\é=13
\defé{\'e}
\catcode`\à=13
\defà{\`a}
\catcode`\è=13
\defè{\`e}
\catcode`\â=13
\defâ{\^a}
\catcode`\ù=13
\defù{\`u}
\catcode`\ê=13
\defê{\^e}
\catcode`\î=13
\defî{\^\i}
\catcode`\ô=13
\defô{\^o}
\newcommand{\field}[1]{\mathbb{#1}}

\newcommand{\Z }{\field{Z}}

\newcommand{\cqf}
{%
\mbox{}%
\nolinebreak%
\hfill%
\rule{2mm}{2mm}%
\medbreak%
\par\noindent%
}

\renewenvironment{proof}{\noindent{\bf Proof}}{\hfill\cqf}

\title{\bf Duality for Modules and Applications to\\ Decoding Linear Codes over Finite Commutative Rings}
\author{Asmae Drhima$^{1}$ and Mustapha Najmeddine$^{2}$}

\begin{document}
\maketitle


\begin{center}
\small{$^{1}$ Department of Mathematics, Faculty of Science, My Isma\"{\i}l University, Meknes, Morocco.
\\E-mail: drhima.asmae@gmail.com
\\$^{2}$ Department of Mathematics, ENSAM, My Isma\"{\i}l University, Meknes, Morocco.
\\E-mail: najmeddine.mustapha@gmail.com}

\bigskip

\end{center}

\begin{abstract}
Using linear functional-based duality of modules, we generalize the syndrome decoding algorithm of linear codes over finite fields to those over finite commutative rings. Moreover, If the ring is local the algorithm is simplified by introducing the control matrix.
\end{abstract}

\bigskip

\bigskip
\section{Introduction}
\hspace{0.5cm}Syndrome decoding is a more efficient method of decoding linear codes over finite fields over a noisy channel [5]. Thus, in this paper we investigate the generalization of the syndrome decoding to linear codes over finite commutative rings. A first generalization was given in [1] via Pontryagin duality. In the same direction we give another generalization using linear functional-based duality. In general, linear functional-based duality and character-based (or Pontryagin) duality are not equivalent (for more details see [8]).

Syndrome decoding of linear codes over finite fields is based on the two following famous results in linear algebra [9] :
\begin{equation}
 C^{\perp\perp} = C
\end{equation}

\begin{equation}
dim(C^{\perp}) = n-dim(C).
\end{equation}
where $C$ is a subspace of $K^{n}$ and $K$ is a field. These properties are not always valid in $A^{n}$ with $A$ is a ring. Wood in [10] has shown the property (1) for any submodule of $A^{n}$ with $A$ is a finite quasi-Frobenius ring. Afterwards, Mittelholzer in [8] extends the class of rings for which the property (1) holds for projective submodules of $A^{n}$ to
artinian rings. In this work we present a detailed proof of (1) for free submodules of $A^{n}$ with $A$ is a
finite ring. In first, we prove in proposition \ref{p3} the property (1) on local finite ring using the existence of free direct summand of free submodule of $A^{n}$. The decomposition of any finite commutative ring as a direct sum of local rings allows us to generalize this property to any finite ring in theorem \ref{t1}.

This article is organized as follows. Section 2 begins by recalls the notions of dual module, orthogonal and bi-orthogonal of submodules in the framework of linear functional-based duality. the following of this section is devoted on the proof of (1). The coding theory begins in section 3 with a review of essential definitions of linear codes over rings. After introducing the concept of dual code, we prove that every dual code of a free code over local ring is also free and its rank satisfies the property (2). Based on results of previous sections especially on theorem \ref{p32}, we present in section 4 the syndrome decoding algorithm. Computing the syndrome is simplified by introducing a control matrix for linear code over local ring.
\bigskip
\section{Duality - Orthogonality}
\bigskip
Throughout this paper, $A$ denotes a finite commutative ring with identity and $M$ an $A$-module.
\begin{defn}
The $A$-module $Hom_{A}(M,A)$ of linear functionals of $M$ is called the dual module of $M$ and denoted $M^{*}.$
\end{defn}

\begin{prop}[\cite{2}, Proposition 6.1.5]
If $M$ is a free module of finite rank $n$, i.e., $M \cong A^{n}$ as $A$-modules. Then $M^{*}$ is free of finite rank $n$ too.
\end{prop}

\begin{defn}
Let $N$ be a submodule of $M$.
\begin{enumerate}
\item The orthogonal of $N$ is the submodule of $M^{*}$ :
$$N^{\circ} = \{f \in M^{*} : f(x) = 0, \forall x \in N\}$$
\item The bi-orthogonal of $N$ is the submodule of $M$ :
$$N^{\circ\circ} = \{x \in M : f(x) = 0, \forall f \in N^{\circ}\}$$
\end{enumerate}
\end{defn}

\begin{prop}[\cite{9}, Proposition 8.7.8]
If $N$ is a submodule of $M$, then the $A$-modules $\left(M/\raisebox{-4pt}{$N$}\right)^{*}$ and $N^{\circ}$ are isomorphic.
\end{prop}

The aim of the following is  to show that every free submodule $N$ of $A^{n}$ satisfies $N^{\circ\circ} = N$. We begin by establishing this result on
a local ring using the following lemma. This last, appears in Appendix II of \cite{6}, is valid on artinian rings, in particular on finite rings.

\begin{lem}\label{l21}
If $A$ is local and $F$ is a free submodule of $A^{n}$. Then there exists a free submodule $Q$ of $A^{n}$ such that $F \oplus Q = A^{n}$
and $A^{n}/\raisebox{-4pt}{$F$}$ is free.
\end{lem}

\begin{prop}\label{p3}
Suppose that $A$ is local. Let $N$ be a free submodule of $A^{n}$ and $x \in A^{n}$.
\begin{enumerate}
\item $x = 0$ if and only if for all $f \in (A^{n})^{*},\ f(x) = 0$.
\item $x \in N$ if and only if for all $f \in N^{\circ},\ f(x) = 0$.
\end{enumerate}
\end{prop}

\begin{proof}
\begin{enumerate}
\item The necessary condition is trivial. Conversely, suppose that $\forall f\in (A^{n})^{*}$ $f(x) = 0$. Let $(e_{i})_{1 \leq i \leq n}$ be a
basis of $A^{n}$ and $(e_{i}^{*})_{1 \leq i \leq n}$ its dual basis in $(A^{n})^{*}$. If $x = \sum\limits_{i=1}^{n} x_{i}e_{i}$ with $(x_{i})_{1 \leq i \leq
n} \in A^{n}$, then $e_{i}^{*}(x) = x_{i} = 0$ $\forall 1 \leq i \leq n$. Therefore $x = 0$.
\item The necessary condition is a consequence of the orthogonal of $N$. Conversely, suppose that $\forall f \in N^{\circ}$ $f(x) = 0$. Let $\varphi \in \left(A^{n}/\raisebox{-4pt}{$N$}\right)^{*}$, the map $f: A^{n} \rightarrow A$, $y \mapsto f(y) = \varphi(\bar{y})$ is linear and $f \in N^{\circ}$. Then $f(x) = 0$. Therefore
$\varphi(\bar{x}) = 0$ for all $\varphi \in \left(A^{n}/\raisebox{-4pt}{$N$}\right)^{*}$ and $A^{n}/\raisebox{-4pt}{$N$}$ is free by Lemma 1. Thus $\bar{x} = \bar{0}$ and $x \in N$. Consequently,
$N^{\circ\circ} = N$.\vspace{-.6cm}
\end{enumerate}
\end{proof}

The ring $A$ is finite commutative. According to the structure theorem for such rings \cite{7}, $A$ can be written as a finite direct sum of local rings
$A_{i}$ ,i.e.,
\begin{equation}
A = \bigoplus\limits_{i=1}^{l} A_{i}.
\end{equation}
where $A_{i} \cong Ae_{i}$ $\forall i = 1,...,l$ and $(e_{i})_{1 \leq i \leq l}$ is a complete system of orthogonal idempotents of $A$, i.e.,
$$e_{i}^{2} = e_{i} \mbox{ , } e_{i}e_{j} = 0 \mbox{ for }i \neq j\mbox{ and }\sum\limits_{i=1}^{l} e_{i} = 1.$$
\\For all $i = 1,...,l,\ M_{_i}=e_{i}M$ is a submodule of $M$ that can provide a structure of $A_{i}$-module \cite{3}. Moreover,
\begin{equation}
M = \bigoplus\limits_{i=1}^{l} M_{i}.
\end{equation}


\begin{lem}\label{l2}
Suppose that $M$ is a finitely generated module over $A$. If $M$ is $A$-free then $M_{i}$ is $A_{i}$-free for all $i = 1,...,l$.
\end{lem}

\begin{proof}

\noindent
Suppose that $M$ is free over $A$. Let $(s_{1},...,s_{n})$ be a basis of $M$. We show that $(e_{i}s_{1},...,e_{i}s_{n})$ is a basis of $M_{i}$.
\\$\bullet$ Let $e_{i}x \in M_{i}$ such that $x = \sum\limits_{j=1}^{n} x_{j}s_{j} \in M$ and $(x_{1},...,x_{n}) \in A^{n}$. $e_{i}x = \sum\limits_{j=1}^{n} (x_{j}e_{i})s_{j} = \sum\limits_{j=1}^{n} (x_{j}e_{i})(e_{i}s_{j}) = \sum\limits_{j=1}^{n} \alpha_{j}(e_{i}s_{j})$ with $\alpha_{j} \in A_{i}$. Therefore $(e_{i}s_{j})_{1 \leq j \leq n}$ generates $M_{i}$ .
\\$\bullet$ Let $(\alpha_{1},...,\alpha_{n}) \in A_{i}^{n}$ such that $\sum\limits_{j=1}^{n} \alpha_{j}(e_{i}s_{j}) = 0 $. Then
$\forall j = 1,...,n$, $\alpha_{j} = a_{j}e_{i} \in Ae_{i}$ and $\sum\limits_{j=1}^{n} (a_{j}e_{i})s_{j} = 0$. Since $(s_{i})_{1 \leq i \leq n}$ is free over
$A$, then $a_{j}e_{i} = \alpha_{j} = 0$ $\forall j$. So $(e_{i}s_{j})_{1 \leq j \leq n}$ is linearly independent.
\end{proof}

Let $f \in Hom_{A}(M,A)$. For all $i = 1,...,l,$ the map $f_{i}: e_{i}x \mapsto e_{i}f(x)$ of $M_{i}$ to $A_{i}$ is $A_{i}$-linear. Furthermore, the map
\begin{equation}\label{bij}
\begin{array}{rcl}Hom_{A}(M,A)& \longrightarrow& \bigoplus\limits_{i=1}^{l} Hom_{A_{i}}(M_{i},A_{i})\\
f &\longmapsto &(f_{1},...,f_{l})
\end{array}
\end{equation}
is bijective. Indeed, if $(f_{1},...,f_{l}) \in \bigoplus\limits_{i=1}^{l} Hom_{A_{i}}(M_{i},A_{i})$. Then $g = \sum \limits_{i=1}^{l} f_{i} \circ pr_{i}$, where
$pr_{i}: M \rightarrow M_{i},$ $x \mapsto e_{i}x$ is the canonical projection, is a linear functional of $M$ because if $a \in A$ and $x \in M$ then
$$\begin{array}{ccl}
g(ax) &=& \sum \limits_{i=1}^{l} f_{_i}\left(e_{_i}ax\right)\\
      &=& \sum \limits_{i=1}^{l} ae_{i}f_{_i}\left(e_{i}x\right)\\
      &=& a \sum \limits_{i=1}^{l} f_{_i}\left(e_{_i}x\right)\\
      &=& ag(x).
\end{array}$$
And $g$  is the unique element of $Hom_{A}(M,A)$ satisfying for all $x \in M$ $f_{i}(e_{i}x) = e_{i} g(x)$.
\begin{thm}\label{t1}
Suppose that $M \cong A^{n}$. If $N$ is a free submodule of $M$, then $N^{\circ\circ} = N$.
\end{thm}
\begin{proof}

\noindent
$N$ is on the form $\bigoplus\limits_{i=1}^{l} N_{i}$ with $N_{i} = e_{i}N$. For all $i = 1,...,l$, $N_{i}$ is a submodule of the
$A_{i}$-module $M_{i}$. Let $f \in Hom_{A}(M,A)$. Using (5), we have $f\mid_{N} = 0$ iff $f_{i}\mid_{N_{i}} = 0$ $\forall i = 1,...,l$. Then
$N^{\circ} = \bigoplus\limits_{i=1}^{l} N_{i}^{\circ}$. Thus $N^{\circ\circ} = \bigoplus\limits_{i=1}^{l} N_{i}^{\circ\circ}$. By Lemma \ref{l2}, $N_{i}$
is free over the local ring $A_{i}$. Then by Proposition \ref{p3}, $N_{i}^{\circ\circ} = N$. Therefore $N^{\circ\circ} = \bigoplus\limits_{i=1}^{l} N_{i} =
N$.
\end{proof}

\section{Linear Codes}
\bigskip
\begin{defn}
\begin{enumerate}\item A linear code $C$ over $A$ of length $n$ is a submodule of $A^{n}$. If $C$ is free over $A$ of rank $k$, $C$ is said an $(n,k)$-code
over $A$. The elements of $C$ are called codewords.
\item Let $C$ be an $(n,k)$-code over $A$. The matrix $G \in M_{k,n}(A)$ whose rows form a basis of $C$ is said to be a generator matrix of $C$.
\item The Hamming distance between $x = (x_{1},...,x_{n})$ and $y = (y_{1},...,y_{n})$ in $A^{n}$ is
$$d(x,y)=|\{ i \in \{ 1,...,n\} : x_{i} \neq y_{i}\} |.$$
\item  The Hamming weight of $x=\left(x_{_1},\ldots,x_{_n}\right) \in A^{n}$ is
$$w(x) = d(x,0)=|\left\{1\le i\le n:x_{_i}\neq0\right\}|.$$
\item The minimal distance of a linear code $C$ is :
$$d(C) = min \{ d(x,y) : x \neq y \in C\} = min \{ w(x) : x \in C - \{ 0 \} \}$$.
\end{enumerate}
\end{defn}

The space $A^{^n}$ with the Hamming distance is a metric space.

\begin{defn}[Dual Code]
Let $C$ be a linear code over $A$ of length $n$. We define the dual code of $C$ by :
$$C^{\perp} = \{ y \in A^{n} : <x,y> = 0, \forall x \in C\}$$
where $<,>$ is the symmetric bilinear form defined for all $x = (x_{1},...,x_{n})$ and
all $y = (y_{1},...,y_{n})$ in $A^{n}$ by :
$$<x,y>=\sum\limits_{i=1}^{n} x_{i}y_{i}.$$
\end{defn}

\begin{prop}\label{p31}
Let $C$ be a linear code over $A$ of length $n$. Then $C^{\perp}$ and $C^{\circ}$ are isomorphic.
\end{prop}

\begin{proof}

\noindent Let
$$\begin{array}{ccccl}
   \varphi &:& A^{n}& \longrightarrow&(A^{n})^{*}\\
   &&x& \longmapsto &\varphi(x)
  \end{array}$$
with $\varphi(x)(y) = <x,y>$ for all $y\in A^{n}$. We show
that $\varphi$ is an isomorphism.
\begin{itemize}
\item It is easy to check that $\varphi$ is linear.
\item $\varphi$ is surjective : Let $f \in (A^{n})^{*}$ and let $(e_{1},...,e_{n})$ be the canonical basis of $A^{n}$. If $y = \sum\limits_{i=1}^{n}
y_{i}e_{i} \in A^{n}$, then $f(y) = \sum\limits_{i=1}^{n} y_{i}f(e_{i}) = <x,y>$ with $x = \sum\limits_{i=1}^{n} f(e_{i})e_{i}$. Therefore $f = <x,.> = \varphi(x)$.
\item $\varphi$ is injective : If $x \in Ker\varphi$. Then for all $y\in A^{n},\ \sum\limits_{i=1}^{n} x_{i}y_{i} = 0$. Especially for
$y = e_{i}$ we have $x_{i} = 0,$ for all $i = 1,\ldots,n$. Therefore $x = 0.$
\end{itemize}
$\varphi$ induces an isomorphism $C^{\perp} \cong \varphi(C^{\perp}) = C^{\circ}$

\end{proof}
\begin{thm}\label{p32}
Let $C$ be an $(n,k)$-code over $A$.
\begin{enumerate}
\item $C^{\perp\perp} = C$.
\item If $A$ is local, then $C^{\perp}$ is free of rank $n-k$.
\end{enumerate}
\end{thm}

\begin{proof}

\noindent
\begin{enumerate}
\item It is clear that $C \subseteq C^{\perp\perp}$. Let $x \in C^{\perp\perp}$, then for all $y \in C^{\perp},\ <x,y> = 0$. By the previous isomorphism between $C^{\perp}$ and $C^{\circ}$ we have for all $f\in C^{\circ}, f(x) = 0$. Thus $x \in C^{\circ\circ} = C,$ by Theorem 1.
\item By Proposition \ref{p31}, we have $C^{\perp} \cong C^{\circ}$  and by Proposition 2, we have $C^{\circ} \cong \left(A^{^n}\!\!\!\raisebox{-3pt}{$/$}\raisebox{-6pt}{$C$}\right)^{*}$. The ring $A$ is
local then by Lemma \ref{l21}, we have that $A^{n}/C$ is free and $rank\left(A^{^n}\!\!\!\raisebox{-3pt}{$/$}\raisebox{-6pt}{$C$}\right) = n-k$. Therefore
$\left(A^{^n}\!\!\!\raisebox{-3pt}{$/$}\raisebox{-6pt}{$C$}\right)^{*}$ is free and $rank\left(A^{^n}\!\!\!\raisebox{-3pt}{$/$}\raisebox{-6pt}{$C$}\right)^{*}
=
rank\left(A^{^n}\!\!\!\raisebox{-3pt}{$/$}\raisebox{-6pt}{$C$}\right)$ by Proposition 1. Thus $C^{\perp}$ is free and $rank(C^{\perp}) = n-k$.
\end{enumerate}\vspace{-.8cm}

\end{proof}

\section{Decoding linear codes}

\subsection{Syndrome decoding}
The principle of this method is to associate each received word after transmission, a quantity $S$ called syndrome. If the error is lightweight this one is
uniquely determined by $S$.\\
Let $C$ be an $(n,k)$-code over $A$, of minimal distance $d$ and $t = \left[\dfrac{d-1}{2}\right]$ the correction capacity of $C.$ Hence $C$ can
detect $(d-1)-$errors and correct $t-$errors [5].

\begin{defn}
Let $x \in A^{n}$. The syndrome of $x$ is
$$S(x) = (<x,y>)_{_{y \in C^{\perp}}}.$$
\end{defn}
We note that the map $S$ is additive. For all $x$ and $y$ in $A^n,\ S(x+y)=S(x)+S(y).$\\

The following result, which generalize the similar fact on fields, is the main tool that allows the code to detect errors.
\begin{prop}
Let $x \in A^{n}$. Then $x \in C$ if and only if $S(x) = 0$.
\end{prop}

\begin{proof}

\noindent
Suppose that $x \in C.$ Then $x\in C^{^{\perp\perp}}$ and for all $y\in C^{\perp},\ <x,y> = 0,$ which implies that $S(x) = 0.$\\
Conversely, if $S(x) = 0,$ then for all $y\in C^{\perp},\ <x,y> = 0.$ Hence $x \in C^{\perp\perp} = C$ by Theorem \ref{p32}.
\end{proof}
As in the case of linear codes on fields, two vectors have the same syndrome if and only if they have the same coset modulo $C:$
\begin{prop}
Let $x,y \in A^{n}$. Then $\bar{x} = \bar{y}$ in $A^{n}\!\!\raisebox{-2pt}{/}\raisebox{-4pt}{$C$}$ if and only if $S(x) = S(y)$.
\end{prop}
\begin{proof}
$$\begin{array}{ccl}
\bar{x} = \bar{y}&\Longleftrightarrow& x-y \in C\\
                 &\Longleftrightarrow& S(x-y) = 0\\
                 &\Longleftrightarrow& S(x) - S(y) = 0\\
                 &\Longleftrightarrow& S(x) = S(y).
  \end{array}$$\vspace{-10pt}
\end{proof}

\begin{cor}
If $r$ is the received word and $e$ the associated error vector. Then $c=r-e \in C$ and $S(r) = S(e)$.
\end{cor}

\begin{prop}
Let $e$ be the error vector. If $w(e) \leq t$ and $S(e) = S$. Then $e$ is the unique vector of weight $\leq t$ having syndrome $S$.
\end{prop}

\begin{proof}

\noindent
Let $e' \in A^{n}$ such that $w(e') \leq t$ and $S(e) = S(e')$ then $e-e' \in C$.
$$\begin{array}{ccl}w(e-e') &=& d(e-e',0)\\
       &=& d(e,e')\\
       &\leq& d(e,0) + d(0,e')\\
       &\leq& w(e) + w(e')\\
        &\leq& 2t < d
   \end{array}$$
\\So $e-e'= 0$ and $e = e'$.\\
\end{proof}

We are, now, able to present the syndrome decoding algorithm.
\newpage
\begin{algorithm}
\caption{The Syndrome Decoding Algorithm}
\begin{algorithmic}[1]
\State Compute the syndrome $S$ of the received word $r$
\State Compute the syndrome of all vectors $e$ of weight $\leq t$
\If {there exists no vector $e$ of weight $\le t$ and syndrome $S$}
\State The algorithm fails\hspace{2cm} \begin{minipage}{5cm}\begin{verbatim}/*more than $t-$errors occur*/\end{verbatim}\end{minipage}
\Else
\State Determine the unique vector $e$ of weight $\leq t$ with syndrome $S$
\State Decode $r$ by $c=r-e \in C$
\EndIf
\end{algorithmic}
\end{algorithm}

\subsection{Control Matrix}
Throughout this part we assume that $A$ is local. Let $C$ be an $(n,k)$-code over $A$ of generator matrix $G$. By Theorem \ref{p32}, we have that
$C^{\perp}$
is free and $rank(C^{\perp}) = n-k$.

\begin{defn}
A generator matrix of $C^{\perp}$ is called control matrix of $C$.
\end{defn}

Let $H$ be a control matrix of $C$. The row vectors of $G$ form a basis of $C$ and the column vectors of $H^{t}$ form a basis of $C^{\perp}$, thus $GH^{t} =
0$. The following theorem gives a necessary and sufficient condition so that a matrix $H \in M_{n-k,n}(A)$ is a control matrix of $C$.

\begin{thm}
Let $H \in M_{n-k,n}(A)$. Then $H$ is a control matrix of $C$ if and only if $GH^{t} = 0$ and row vectors of $H$ are linearly independent.
\end{thm}

\begin{proof}

\noindent
The necessary condition is a consequence of the definition of control matrix. Conversely, suppose that $GH^{t}=0$ and $(e_{1},...,e_{n-k})$,
the row vectors of $H$, are linearly independent. Then for all $1\leq i \leq n-k,\ e_{_i} \in C^{^{\perp}}$. Hence $\bigoplus\limits_{i=1}^{n-k} Ae_{i} \subseteq
C^{\perp}$ and $rg(\bigoplus\limits_{i=1}^{n-k} Ae_{i}) = n-k = rg(C^{\perp})$. Therefore $\bigoplus\limits_{i=1}^{n-k} Ae_{i} \cong C^{\perp}$ and
$\mid\bigoplus\limits_{i=1}^{n-k} Ae_{i}\mid = \mid C^{\perp} \mid$. Thus $\bigoplus\limits_{i=1}^{n-k} Ae_{i} = C^{\perp}$ and $(e_{i})_{1\leq i \leq n-k}$
is a basis of $C^{\perp}.$
\end{proof}
\begin{cor}
If the generator matrix of $C$ is in standard form (i.e) $G = (I_{k},P)$. Then $H = (-P^{t},I_{n-k})$ is a control matrix of $C$, with
$P \in M_{k,n-k}(A)$ and $I_{k}$ denotes the identity matrix of order $k$.
\end{cor}

\begin{proof}

\noindent
We have $GH^{t} = -P + P = 0$ and it is clear that the lines of $H$ are linearly independent.\vspace{-.7cm}
\end{proof}

\begin{prop}
Let $H$ be a matrix control of $C$. Then $\forall x,y \in A^{n}$ :
\begin{enumerate}
\item $S(x) = 0$ if and only if $Hx^{t} = 0$
\item $S(x) = S(y)$ if and only if $Hx^{t} = Hy^{t}$.
\end{enumerate}
\end{prop}

\begin{proof}

\noindent
Let
$$H = \left(
                                      \begin{array}{cccc}
                                        e_{_{11}} &\cdots&\cdots  & e_{_{1n}} \\
                                         \vdots&  &  &\vdots  \\
                                        e_{_{n-k,1}} &\cdots & \cdots & e_{_{n-k,n}} \\
                                      \end{array}
                                    \right)$$
be a control matrix of $C$, with $e_{i} = (e_{i1},...,e_{in})$ in $C^{\perp}$ for all $i = 1,...,n-k$ \\
If $x = (x_{1},...,x_{n}) \in A^{n}$ then $Hx^{t} = (\sum\limits_{j=1}^{n} x_{j}e_{ij})_{1\leq i \leq n-k} = (<x,e_{i}>)_{1\leq i \leq n-k}$.
\begin{enumerate}
 \item\label{1} If $S(x) = 0$ then for all $y\in C^{\perp},\ <x,y> = 0.$ Hence for all $i=1,\ldots,n-k,\ \left<x,e_{_i}\right> = 0$ and  $Hx^{t} = 0.$\\
Conversely, suppose that $Hx^{t}=0$. Let $z = \sum\limits_{i=1}^{n-k} a_{i}e_{i} \in C^{\perp}$ Then
$$\begin{array}{ccl}
<x,z> &=&\left<x,\sum\limits_{i=1}^{n-k} a_{i}e_{_i}\right>\\
      &=& \sum\limits_{i=1}^{n-k}a_{i}<x,e_{i}> = 0.
  \end{array}$$
Therefore $S(x) = 0$.
\item Straightforward from $(\ref{1}).$
\end{enumerate}\vspace{-.7cm}
\end{proof}

The minimal distance of a code $C$ is an important factor in the decoding algorithm of linear codes. It allows us to determine the correction capability of
the code. The following proposition gives us a way to find the minimal distance by using the control matrix.

\begin{prop}[\cite{1}, Proposition 7]\label{p46}
Let $H$ be a control matrix of $C$. Then the minimal distance of $C$ is the minimal number of dependent columns of $H$.
\end{prop}

\begin{exmp}
We give an example illustrating the concepts studied above. The computations are simple but tedious, therefore we have used the computer algebra system Maple
to verify the calculations [4].\\
Let $A$ be the finite commutative local ring $\Z/4\Z$. Let $C$ be the $(20,10)$-linear code over $A$ of generator matrix
\begin{center}
  $G = (I_{10},P)$
\end{center}
with
\begin{center}
  $P = \left(
         \begin{array}{cccccccccc}
           1 & 0 & 3 & 0 & 1 & 3 & 0 & 2 & 2 & 0 \\
           1 & 3 & 0 & 3 & 2 & 1 & 1 & 3 & 3 & 3 \\
           3 & 1 & 2 & 0 & 1 & 1 & 3 & 2 & 3 & 0 \\
           2 & 0 & 2 & 2 & 2 & 3 & 3 & 3 & 3 & 3 \\
           0 & 0 & 3 & 0 & 0 & 0 & 2 & 0 & 0 & 0 \\
           3 & 1 & 3 & 2 & 3 & 3 & 3 & 1 & 2 & 2 \\
           2 & 0 & 0 & 1 & 2 & 1 & 0 & 1 & 1 & 2 \\
           2 & 0 & 1 & 3 & 1 & 1 & 1 & 0 & 3 & 1 \\
           0 & 2 & 1 & 1 & 2 & 2 & 1 & 3 & 0 & 3 \\
           0 & 0 & 1 & 2 & 1 & 2 & 2 & 0 & 1 & 1 \\
         \end{array}
       \right)$
\end{center}
We begin by computing the control matrix $H$ for the code $C$.
\begin{verbatim}
P := <<1,1,3,2,0,3,2,2,0,0> | <0,3,1,0,0,1,0,0,2,0> |
      <3,0,2,2,3,3,0,1,1,1> | <0,3,0,2,0,2,1,3,1,2> |
      <1,2,1,2,0,3,2,1,2,1> | <3,1,1,3,0,3,1,1,2,2> |
      <0,1,3,3,2,3,0,1,1,2> | <2,3,2,3,0,1,1,0,3,0> |
      <2,3,3,3,0,2,1,3,0,1> | <0,3,0,3,0,2,2,1,3,1>>;
H := <-Transpose(P)| IdentityMatrix(10)> mod 4:
evalm(H);

\end{verbatim}
\begin{center}
 $ H := \left(
 \begin{array}{cccccccccccccccccccc}
         3 & 3 & 1 & 2 & 0 & 1 & 2 & 2 & 0 & 0 & 1 & 0 & 0 & 0 & 0 & 0 & 0 & 0 & 0 & 0 \\
         0 & 1 & 3 & 0 & 0 & 3 & 0 & 0 & 2 & 0 & 0 & 1 & 0 & 0 & 0 & 0 & 0 & 0 & 0 & 0 \\
         1 & 0 & 2 & 2 & 1 & 1 & 0 & 3 & 3 & 3 & 0 & 0 & 1 & 0 & 0 & 0 & 0 & 0 & 0 & 0 \\
         0 & 1 & 0 & 2 & 0 & 2 & 3 & 1 & 3 & 2 & 0 & 0 & 0 & 1 & 0 & 0 & 0 & 0 & 0 & 0 \\
         3 & 2 & 3 & 2 & 0 & 1 & 2 & 3 & 2 & 3 & 0 & 0 & 0 & 0 & 1 & 0 & 0 & 0 & 0 & 0 \\
         1 & 3 & 3 & 1 & 0 & 1 & 3 & 3 & 2 & 2 & 0 & 0 & 0 & 0 & 0 & 1 & 0 & 0 & 0 & 0 \\
         0 & 3 & 1 & 1 & 2 & 1 & 0 & 3 & 3 & 2 & 0 & 0 & 0 & 0 & 0 & 0 & 1 & 0 & 0 & 0 \\
         2 & 1 & 2 & 1 & 0 & 3 & 3 & 0 & 1 & 0 & 0 & 0 & 0 & 0 & 0 & 0 & 0 & 1 & 0 & 0 \\
         2 & 1 & 1 & 1 & 0 & 2 & 3 & 1 & 0 & 3 & 0 & 0 & 0 & 0 & 0 & 0 & 0 & 0 & 1 & 0 \\
         0 & 1 & 0 & 1 & 0 & 2 & 2 & 3 & 1 & 3 & 0 & 0 & 0 & 0 & 0 & 0 & 0 & 0 & 0 & 1
       \end{array}
       \right)$
\end{center}
By proposition \ref{p46}, the minimal distance of $C$ is $d = 3$ and this permits to detect $2$ errors and correct $t =\left[\dfrac{d-1}{2}\right]=1$ error.
\\Let $c = 10202230013001002303
  \in C$ the transmitted codeword and
$r = 10202130013001002303$ the received noisy word. To determine if an error exists in this received word, we compute the syndrome of $r$.
\\Firstly, we define a procedure that returns the syndrome of a vector
\begin{verbatim}
Syndrome := proc(A,x):
return MatrixVectorMultiply(A,x) mod 4;
end proc:
\end{verbatim}
We compute the syndrome of $r$
\begin{verbatim}
Syndrome(H,r);
\end{verbatim}
\begin{center}
$ \left(
    \begin{array}{c}
      3 \\
      1 \\
      3 \\
      2 \\
      3 \\
      3 \\
      3 \\
      1 \\
      2 \\
      2 \\
    \end{array}
  \right)$
\end{center}
Because this syndrome is nonzero, we know $r$ is erroneous. To find the error in $r$ we must Compute the syndrome of all vectors of $A^{20}$ of weight $1$ and compare them with the syndrome of $r$. Therefore, We collect these vectors in a matrix $E$.
\begin{verbatim}
E := <IdentityMatrix(20) | 2*IdentityMatrix(20) | 3*IdentityMatrix(20)>:
evalm(E);
\end{verbatim}
$$E=\left(\arraycolsep=1.1pt
         \begin{array}{cccccccccccccccccccccccccccccccccccccccccccccccccccccccccccc}
           1 & 0 & 0 & 0 & 0 & 0 & 0 & 0 & 0 & 0 & 0 & 0 & 0 & 0 & 0 & 0 & 0 & 0 & 0 & 0 & 2 & 0 & 0 & 0 & 0 & 0 & 0 & 0 & 0 & 0 & 0 & 0 & 0 & 0 & 0 & 0 & 0 & 0 & 0 & 0 & 3 & 0 & 0 & 0 & 0 & 0 & 0 & 0 & 0 & 0 & 0 & 0 & 0 & 0 & 0 & 0 & 0 & 0 & 0 & 0 \\
           0 & 1 & 0 & 0 & 0 & 0 & 0 & 0 & 0 & 0 & 0 & 0 & 0 & 0 & 0 & 0 & 0 & 0 & 0 & 0 & 0 & 2 & 0 & 0 & 0 & 0 & 0 & 0 & 0 & 0 & 0 & 0 & 0 & 0 & 0 & 0 & 0 & 0 & 0 & 0 & 0 & 3 & 0 & 0 & 0 & 0 & 0 & 0 & 0 & 0 & 0 & 0 & 0 & 0 & 0 & 0 & 0 & 0 & 0 & 0 \\
           0 & 0 & 1 & 0 & 0 & 0 & 0 & 0 & 0 & 0 & 0 & 0 & 0 & 0 & 0 & 0 & 0 & 0 & 0 & 0 & 0 & 0 & 2 & 0 & 0 & 0 & 0 & 0 & 0 & 0 & 0 & 0 & 0 & 0 & 0 & 0 & 0 & 0 & 0 & 0 & 0 & 0 & 3 & 0 & 0 & 0 & 0 & 0 & 0 & 0 & 0 & 0 & 0 & 0 & 0 & 0 & 0 & 0 & 0 & 0 \\
           0 & 0 & 0 & 1 & 0 & 0 & 0 & 0 & 0 & 0 & 0 & 0 & 0 & 0 & 0 & 0 & 0 & 0 & 0 & 0 & 0 & 0 & 0 & 2 & 0 & 0 & 0 & 0 & 0 & 0 & 0 & 0 & 0 & 0 & 0 & 0 & 0 & 0 & 0 & 0 & 0 & 0 & 0 & 3 & 0 & 0 & 0 & 0 & 0 & 0 & 0 & 0 & 0 & 0 & 0 & 0 & 0 & 0 & 0 & 0 \\
           0 & 0 & 0 & 0 & 1 & 0 & 0 & 0 & 0 & 0 & 0 & 0 & 0 & 0 & 0 & 0 & 0 & 0 & 0 & 0 & 0 & 0 & 0 & 0 & 2 & 0 & 0 & 0 & 0 & 0 & 0 & 0 & 0 & 0 & 0 & 0 & 0 & 0 & 0 & 0 & 0 & 0 & 0 & 0 & 3 & 0 & 0 & 0 & 0 & 0 & 0 & 0 & 0 & 0 & 0 & 0 & 0 & 0 & 0 & 0 \\
           0 & 0 & 0 & 0 & 0 & 1 & 0 & 0 & 0 & 0 & 0 & 0 & 0 & 0 & 0 & 0 & 0 & 0 & 0 & 0 & 0 & 0 & 0 & 0 & 0 & 2 & 0 & 0 & 0 & 0 & 0 & 0 & 0 & 0 & 0 & 0 & 0 & 0 & 0 & 0 & 0 & 0 & 0 & 0 & 0 & 3 & 0 & 0 & 0 & 0 & 0 & 0 & 0 & 0 & 0 & 0 & 0 & 0 & 0 & 0 \\
           0 & 0 & 0 & 0 & 0 & 0 & 1 & 0 & 0 & 0 & 0 & 0 & 0 & 0 & 0 & 0 & 0 & 0 & 0 & 0 & 0 & 0 & 0 & 0 & 0 & 0 & 2 & 0 & 0 & 0 & 0 & 0 & 0 & 0 & 0 & 0 & 0 & 0 & 0 & 0 & 0 & 0 & 0 & 0 & 0 & 0 & 3 & 0 & 0 & 0 & 0 & 0 & 0 & 0 & 0 & 0 & 0 & 0 & 0 & 0 \\
           0 & 0 & 0 & 0 & 0 & 0 & 0 & 1 & 0 & 0 & 0 & 0 & 0 & 0 & 0 & 0 & 0 & 0 & 0 & 0 & 0 & 0 & 0 & 0 & 0 & 0 & 0 & 2 & 0 & 0 & 0 & 0 & 0 & 0 & 0 & 0 & 0 & 0 & 0 & 0 & 0 & 0 & 0 & 0 & 0 & 0 & 0 & 3 & 0 & 0 & 0 & 0 & 0 & 0 & 0 & 0 & 0 & 0 & 0 & 0 \\
           0 & 0 & 0 & 0 & 0 & 0 & 0 & 0 & 1 & 0 & 0 & 0 & 0 & 0 & 0 & 0 & 0 & 0 & 0 & 0 & 0 & 0 & 0 & 0 & 0 & 0 & 0 & 0 & 2 & 0 & 0 & 0 & 0 & 0 & 0 & 0 & 0 & 0 & 0 & 0 & 0 & 0 & 0 & 0 & 0 & 0 & 0 & 0 & 3 & 0 & 0 & 0 & 0 & 0 & 0 & 0 & 0 & 0 & 0 & 0 \\
           0 & 0 & 0 & 0 & 0 & 0 & 0 & 0 & 0 & 1 & 0 & 0 & 0 & 0 & 0 & 0 & 0 & 0 & 0 & 0 & 0 & 0 & 0 & 0 & 0 & 0 & 0 & 0 & 0 & 2 & 0 & 0 & 0 & 0 & 0 & 0 & 0 & 0 & 0 & 0 & 0 & 0 & 0 & 0 & 0 & 0 & 0 & 0 & 0 & 3 & 0 & 0 & 0 & 0 & 0 & 0 & 0 & 0 & 0 & 0 \\
           0 & 0 & 0 & 0 & 0 & 0 & 0 & 0 & 0 & 0 & 1 & 0 & 0 & 0 & 0 & 0 & 0 & 0 & 0 & 0 & 0 & 0 & 0 & 0 & 0 & 0 & 0 & 0 & 0 & 0 & 2 & 0 & 0 & 0 & 0 & 0 & 0 & 0 & 0 & 0 & 0 & 0 & 0 & 0 & 0 & 0 & 0 & 0 & 0 & 0 & 3 & 0 & 0 & 0 & 0 & 0 & 0 & 0 & 0 & 0 \\
           0 & 0 & 0 & 0 & 0 & 0 & 0 & 0 & 0 & 0 & 0 & 1 & 0 & 0 & 0 & 0 & 0 & 0 & 0 & 0 & 0 & 0 & 0 & 0 & 0 & 0 & 0 & 0 & 0 & 0 & 0 & 2 & 0 & 0 & 0 & 0 & 0 & 0 & 0 & 0 & 0 & 0 & 0 & 0 & 0 & 0 & 0 & 0 & 0 & 0 & 0 & 3 & 0 & 0 & 0 & 0 & 0 & 0 & 0 & 0 \\
           0 & 0 & 0 & 0 & 0 & 0 & 0 & 0 & 0 & 0 & 0 & 0 & 1 & 0 & 0 & 0 & 0 & 0 & 0 & 0 & 0 & 0 & 0 & 0 & 0 & 0 & 0 & 0 & 0 & 0 & 0 & 0 & 2 & 0 & 0 & 0 & 0 & 0 & 0 & 0 & 0 & 0 & 0 & 0 & 0 & 0 & 0 & 0 & 0 & 0 & 0 & 0 & 3 & 0 & 0 & 0 & 0 & 0 & 0 & 0 \\
           0 & 0 & 0 & 0 & 0 & 0 & 0 & 0 & 0 & 0 & 0 & 0 & 0 & 1 & 0 & 0 & 0 & 0 & 0 & 0 & 0 & 0 & 0 & 0 & 0 & 0 & 0 & 0 & 0 & 0 & 0 & 0 & 0 & 2 & 0 & 0 & 0 & 0 & 0 & 0 & 0 & 0 & 0 & 0 & 0 & 0 & 0 & 0 & 0 & 0 & 0 & 0 & 0 & 3 & 0 & 0 & 0 & 0 & 0 & 0 \\
           0 & 0 & 0 & 0 & 0 & 0 & 0 & 0 & 0 & 0 & 0 & 0 & 0 & 0 & 1 & 0 & 0 & 0 & 0 & 0 & 0 & 0 & 0 & 0 & 0 & 0 & 0 & 0 & 0 & 0 & 0 & 0 & 0 & 0 & 2 & 0 & 0 & 0 & 0 & 0 & 0 & 0 & 0 & 0 & 0 & 0 & 0 & 0 & 0 & 0 & 0 & 0 & 0 & 0 & 3 & 0 & 0 & 0 & 0 & 0 \\
           0 & 0 & 0 & 0 & 0 & 0 & 0 & 0 & 0 & 0 & 0 & 0 & 0 & 0 & 0 & 1 & 0 & 0 & 0 & 0 & 0 & 0 & 0 & 0 & 0 & 0 & 0 & 0 & 0 & 0 & 0 & 0 & 0 & 0 & 0 & 2 & 0 & 0 & 0 & 0 & 0 & 0 & 0 & 0 & 0 & 0 & 0 & 0 & 0 & 0 & 0 & 0 & 0 & 0 & 0 & 3 & 0 & 0 & 0 & 0 \\
           0 & 0 & 0 & 0 & 0 & 0 & 0 & 0 & 0 & 0 & 0 & 0 & 0 & 0 & 0 & 0 & 1 & 0 & 0 & 0 & 0 & 0 & 0 & 0 & 0 & 0 & 0 & 0 & 0 & 0 & 0 & 0 & 0 & 0 & 0 & 0 & 2 & 0 & 0 & 0 & 0 & 0 & 0 & 0 & 0 & 0 & 0 & 0 & 0 & 0 & 0 & 0 & 0 & 0 & 0 & 0 & 3 & 0 & 0 & 0 \\
           0 & 0 & 0 & 0 & 0 & 0 & 0 & 0 & 0 & 0 & 0 & 0 & 0 & 0 & 0 & 0 & 0 & 1 & 0 & 0 & 0 & 0 & 0 & 0 & 0 & 0 & 0 & 0 & 0 & 0 & 0 & 0 & 0 & 0 & 0 & 0 & 0 & 2 & 0 & 0 & 0 & 0 & 0 & 0 & 0 & 0 & 0 & 0 & 0 & 0 & 0 & 0 & 0 & 0 & 0 & 0 & 0 & 3 & 0 & 0 \\
           0 & 0 & 0 & 0 & 0 & 0 & 0 & 0 & 0 & 0 & 0 & 0 & 0 & 0 & 0 & 0 & 0 & 0 & 1 & 0 & 0 & 0 & 0 & 0 & 0 & 0 & 0 & 0 & 0 & 0 & 0 & 0 & 0 & 0 & 0 & 0 & 0 & 0 & 2 & 0 & 0 & 0 & 0 & 0 & 0 & 0 & 0 & 0 & 0 & 0 & 0 & 0 & 0 & 0 & 0 & 0 & 0 & 0 & 3 & 0 \\
           0 & 0 & 0 & 0 & 0 & 0 & 0 & 0 & 0 & 0 & 0 & 0 & 0 & 0 & 0 & 0 & 0 & 0 & 0 & 1 & 0 & 0 & 0 & 0 & 0 & 0 & 0 & 0 & 0 & 0 & 0 & 0 & 0 & 0 & 0 & 0 & 0 & 0 & 0 & 2 & 0 & 0 & 0 & 0 & 0 & 0 & 0 & 0 & 0 & 0 & 0 & 0 & 0 & 0 & 0 & 0 & 0 & 0 & 0 & 3 \\
         \end{array}\right)$$
We use the following commands to find the column in $E$ that matches the error.
\begin{verbatim}
fc := 0:
cn := 0:
while (fc <> 1) and (cn < 60)do
cn := cn+1;
if Equal(Syndrome(H,Column(E,cn)), Syndrome(H,r)) = true then
fc := 1;
fi;
od:
\end{verbatim}
\begin{verbatim}
cn:
\end{verbatim}
\begin{center}
  $46$\vspace{1cm}
\end{center}

This value for cn indicates that the error is the $46^{^{\mbox{th}}}$ column. We can then see that the error vector that corresponds to $r$ as follows
\begin{verbatim}
error := Column(E,cn):
evalm(error);
\end{verbatim}
\begin{center}
  $\left(
    \begin{array}{cccccccccccccccccccc}
      0 & 0 & 0 & 0 & 0 & 3 & 0 & 0 & 0 & 0 & 0 & 0 & 0 & 0 & 0 & 0 & 0 & 0 & 0 & 0 \\
    \end{array}
  \right)$
\end{center}
The error is $e=00000300000000000000$. Thus the transmitted codeword is $c=r-e=10202230013001002303.$
\end{exmp}

\bigskip
\centerline{\Large\bf Acknowledgments}
\hspace{0.5cm}The authors would like to thank the group of algebra and geometry of Moulay Ismaïl university, especially M. Ait Ben Haddou. Thanks also to M. E.
Charkani from university of Fes for his help in commutative algebra.
\bigskip\bigskip


\bigskip

\begin{thebibliography}{999}\addcontentsline{toc}{section}{\protect\numberline{}{Bibliography}}
\bibitem{1} K. Abdelmoumen, M. Najmeddine et H. Ben-Azza. Pontrjagin Duality and Codes over Finite Commutative Rings. World Academy of Science, Engineering and Technology 56 2011, pp. 999-1003.
\bibitem{2} W.A. Adkins and S.H. Weintraub. Algebra : an Approach via Module Theory. Springer, 1992.
\bibitem{3} N. Bourbaki. Algèbre chapitres 1 à 3. Springer, 1970.
\bibitem{4} R-E. Kilima, N. Sigmon, E. Stitzinger. Applications of abstract algebra with Maple. CRC Press
LLC; 1999.
\bibitem{5} S.Ling and C.Xing. Coding Theory : A First Course. Cambridge University Press, 2004.
\bibitem{6} H.-A. Loeliger and T. Mittelholzer. Convolutional Codes over Groups. IEEE Trans. Information Th., Vol. 42(6), Nov. 1996, pp. 1660-1686.
\bibitem{7} B.R. McDonald. Finite Rings with Identity. Marcel Dekker, 1974.
\bibitem{8} T. Mittelholzer. Linear Codes and Their Duals over Artinian Rings. Codes, Systems, and Graphical Models, Eds. B. Marcus and J. Rosenthal. The IMA Volumes in Mathematics and its Applications, pp. 361-379, Springer, 2001.
\bibitem{9} P. Tauvel. Mathématiques Générales pour l'Agrégation. Masson, Paris, 1992.
\bibitem{10} J.A. Wood. Duality for Modules over Finite Rings and Applications to Coding Theory. American J. of Math., Vol 121.3, June 1999, pp. 555-575.

\end{thebibliography}
\end{document}